\documentclass[runningheads,orivec]{llncs}

\usepackage{packages}

\newcommand{\mihaela}[1]{\todo[color=green!20,linecolor=black!60,size=\scriptsize]{M: #1}}

\makeatletter
 
\def\orcidID#1{\smash{\href{http://orcid.org/#1}{\protect\raisebox{-1.25pt}{\protect\includegraphics{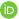}}}}}
\makeatother

\newcommand{\lst}[1]{{\lstinline{#1}}}

\newcommand{\resp}{resp.}
\newcommand{\mypar}[1]{\smallskip\noindent\emph{#1}}
\newcommand{\myfpar}[1]{\noindent\fbox{#1}}

\newcommand{\slah}{\textsf{SLAH}}
\newcommand{\PbA}{\textsf{PA}}
\newcommand{\EPbA}{\textsf{EPA}}
\newcommand{\qfpa}{\textsf{QFPA}}

\newcommand{\NN}{\mathbb{N}}  
\newcommand{\tfun}{\rightarrow}  
\newcommand{\pfun}{\rightharpoonup} 
\newcommand{\qf}{\mathtt{qf}}  
\newcommand{\fv}{\mathtt{fv}}  

\newcommand{\limp}{\rightarrow}

\newcommand{\emp}{\mathtt{emp}}
\newcommand{\pto}{\mapsto}
\newcommand{\sepc}{\ast}
\newcommand{\bsepc}{\scalebox{1.5}{\raisebox{-0.2ex}{$\ast$}}}

\newcommand{\blk}{\mathtt{blk}}

\newcommand{\hls}[1]{\mathtt{hls}^{#1}}

\newcommand{\ls}{\mathtt{ls}}

\newcommand{\abs}{\mathtt{Abs}}

\newcommand{\wpos}[1]{\mathtt{dom}(#1)}

\newcommand{\eub}{\mathtt{EUB}}

\newcommand{\cspen}{\textsc{CompSPEN}}
\newcommand{\cspenp}{\textsc{CompSPEN}$^+$}

\newcommand{\hide}[1]{ }

\newcommand \ltrue {{\tt true}}

\newcommand \isnonemp {{\tt isNonEmp}}

\newcommand \patoms {{\tt PAtoms}}
\newcommand \atoms {{\tt Atoms}}
\newcommand \atomhead {{\tt start}}
\newcommand \atomtail {{\tt end}}

\newcommand{\addr}{\mathcal{A}}

\newcommand{\cV}{\mathcal{V}}

\title{Deciding Separation Logic with Pointer Arithmetic and Inductive Definitions}
\author{%
Wanyun Su\inst{2} \and
Zhilin Wu\inst{2}\orcidID{0000-0003-0899-628X} \and
Mihaela Sighireanu\inst{1}\orcidID{0000-0002-1925-089X}
}
\institute{%
   State Key Laboratory of Computer Science, \\
   Institute of Software, Chinese Academy of Sciences, China \\
\and
   LMF, ENS Paris-Saclay, University Paris-Saclay and CNRS, France
}
\date{}
\authorrunning{W. Su, Z. Wu and M. Sighireanu}

\begin{document}


\maketitle
\sloppy



\begin{abstract}
Pointer arithmetic is widely used in low-level programs, e.g. memory allocators. 
%
The specification of such programs usually requires using pointer arithmetic inside inductive definitions to define the common data structures, e.g. heap lists in memory allocators. 
In this work, we investigate decision problems for SLAH, a separation logic fragment 
that allows pointer arithmetic inside inductive definitions, 
thus enabling specification of properties for programs manipulating heap lists.
%
Pointer arithmetic inside inductive definitions is challenging for automated reasoning.  
We tackle this challenge and achieve decision procedures for both satisfiability and entailment of SLAH formulas. 
The crux of our decision procedure for satisfiability is to compute summaries of inductive definitions. We show that although the summary is naturally expressed as an existentially quantified non-linear arithmetic formula, it can actually be transformed into an equivalent linear arithmetic formula. 
The decision procedure for entailment, on the other hand, has to match and split the spatial atoms according to the arithmetic relation between address variables. 
We report on the implementation of these decision procedures and their 
good performance in solving problems issued from the verification of 
building block programs used in memory allocators.
\hide{
We study the decidability of the verification problem for 
an extension of the array separation logic (ASL) 
allowing the specification of data structures like 
pointers, memory blocks or heap-lists, 
which are lists built inside memory blocks. 
The logic, called \slah, adds to ASL 
inductively defined predicates specifying singly linked heap-lists. 
We show that the logic is suitable for the compositional verification 
of programs manipulating heap-lists, e.g., memory allocators. 
We propose decision procedures for satisfiability and entailment of \slah, 
which are substantial extensions of those proposed for ASL 
because they have to deal with pointer arithmetic inside inductive definitions.
Our decision procedures propose some new ideas to tackle this challenge. 
The crux of our decision procedure for satisfiability is 
to compute a summary of the inductive definition of the heap-list predicate.
We show that although the summary is expressed as 
an existentially quantified non-linear arithmetic constraint, 
it can actually be transformed into an equivalent linear arithmetic formula. 
The decision procedure for entailment 
uses the decision procedure for satisfiability as an oracle 
and deals with the pointer arithmetic in the inductive definitions. 
We implemented the decision procedures and 
did experiments to evaluate their performance. 
The experimental results demonstrate the effectiveness of the solver
for deciding problems issued from the verification 
of heap-list manipulating programs.
}
\end{abstract}



\section{Introduction}

\mypar{Context.}
Separation Logic (SL, \cite{Reynolds:2002,OHearn19}), an extension of Hoare logic, is a well-established formalism for the verification of heap manipulating programs. SL features a \emph{separating conjunction operator} and \emph{inductive predicates}, which allow to express how data structures are laid out in memory in an abstract way. Since its introduction, various verification tools based on separation logic have been developed. A notable one among them is the INFER tool \cite{CD11}, which was acquired by Facebook in 2013 and has been actively used in its development process \cite{CDD+15}. 

Decision procedures for separation logic formulas are vital for the automation of the verification process.
These decision procedures mostly focused on separation logic fragments called \emph{symbolic heaps} (SH)~\cite{BerdineCO04}, since they provide a good compromise between expressivity and tractability. 
The SH fragments comprise existentially quantified formulas that are 
conjunctions of atoms encoding 
aliasing constraints $x = y$ and $x \neq y$ between pointers $x$ and $y$, 
points-to constraints $x\pto v$ expressing that at the address $x$ is stored the value $v$, 
and predicate atoms $P(\vec{x})$ defining unbounded memory regions of a particular structure. The points-to and predicate atoms, also called spatial atoms,
are composed using the separating conjunction to specify the disjointness of memory blocks they specify.

Let us briefly summarize the results on the SH fragments in the sequel.
For the SH fragment with the singly linked list-segment predicate, arguably the simplest SH fragment, its satisfiability and entailment problems have been shown to be in PTIME~\cite{CookHOPW11} and efficient solvers have been developed for it~\cite{SLCOMPsite}.
The situation changes for more complex inductive predicates: 
The satisfiability problem for the SH fragment with a class of general inductive predicates was shown to be EXPTIME-complete~\cite{DBLP:conf/csl/BrotherstonFPG14}. On the other hand, the entailment problem for the SH fragment with slightly less general inductive predicates was shown to be 2-EXPTIME-complete~\cite{KatelaanMZ19,DBLP:conf/lpar/EchenimIP20}. 

\mypar{Motivation.}
Vast majority of the work on the verification of heap manipulating programs based on SL assumes that the addresses are \emph{nominal}, that is, they can be compared with only equality or disequality, but not ordered or obtained by arithmetic operations.
However, pointer arithmetic is widely used in low-level programs 
to access data inside memory blocks. Memory allocators are such low-level programs. 
They assume that the memory is organized into a linked list of memory chunks, 
called heap lists in this paper; pointer arithmetic is
used to jump from a memory chunk to the next one \cite{Knuth97,WJ+95}. 
%
%
There have been some work to use separation logic for the static analysis and deductive verification of these low-level programs~\cite{CalcagnoDOHY06,MartiAY06,Chlipala11}. 
Moreover, researchers have also started to investigate the decision procedures for separation logic fragments containing pointer arithmetic. 
For instance, \emph{Array separation logic} (ASL) was proposed in~\cite{BrotherstonGK17}, which includes pointer arithmetic, the constraints $\blk(x,y)$ denoting a block of memory units from the address $x$ to $y$, as well as the points-to constraints $x\pto v$. It was shown in~\cite{BrotherstonGK17} that for ASL, the satisfiability is NP-complete and the entailment is in coNEXP resp. coNP for quantified resp. quantifier-free formulas.
Furthermore, the decidability can be preserved even if ASL is extended with the list-segment predicate~\cite{DBLP:journals/corr/abs-1802-05935}.
Very recently, Le identified in~\cite{DBLP:conf/vmcai/Le21} two fragments 
of ASL extended with a class of general inductive predicates for which 
the satisfiability (but not entailment) problem is decidable.

Nevertheless, none of the aforementioned work is capable of reasoning about heap lists, or generally speaking, pointer arithmetic inside inductive definitions, in a \emph{sound and complete} way. The state-of-the-art static analysis and verification tools, e.g. \cite{CalcagnoDOHY06,Chlipala11,MartiAY06}, resort to sound  (but incomplete) heuristics or interactive theorem provers, for reasoning about heap lists. 
On the other hand, the decision procedures for ASL or its extensions, e.g. \cite{BrotherstonGK17,DBLP:journals/corr/abs-1802-05935,DBLP:conf/vmcai/Le21}, are unable to tackle heap lists.
This motivates us to raise the following research question: \emph{Can decision procedures be achieved for separation logic fragments allowing pointer arithmetic inside inductive definitions}?

\mypar{Contribution.}
In this work, we propose decision procedures for a fragment of separation logic called {\slah}, which allows pointer arithmetic inside inductive definitions, 
so that inductive predicates specifying heap lists can be defined. 
We consider both satisfiability and entailment problems and show that they are NP-complete and coNP-complete respectively. 
The decision procedure for satisfiability is obtained by computing an equi-satisfiable abstraction in Presburger arithmetic, whose crux is to show that the summaries of the heap list predicates, which are naturally formalized as existentially quantified non-linear arithmetic formulas, can actually be transformed into Presburger arithmetic formulas. 
The decision procedure for entailment, on the other hand, reduces the problem to multiple instances of an ordered entailment problem, where all the address terms of spatial atoms are ordered. 
The ordered entailment problem is then decided by matching each spatial atom in the consequent to some spatial formula obtained from the antecedent by partitioning and splitting the spatial atoms according to the arithmetic relations between address variables.
Splitting a spatial atom into multiple ones in the antecedent is attributed to pointer arithmetic and unnecessary for SH fragments with nominal addresses. 

We implemented the decision procedures on top of \cspen\ solver~\cite{GuCW16}.   
We evaluate the performance of the new solver, called {\cspenp}~\cite{CompSpenSite}, 
on a set of formulas originating from path conditions and verification conditions of programs working on heap lists in memory allocators. 
We also randomly generate some formulas, in order to test the scalability of {\cspenp} further.
The experimental results show that {\cspenp} is able to solve the satisfiability and entailment problems for {\slah} efficiently (in average, less than 1 second for satisfiability and less than 15 seconds for entailment).

To the best of our knowledge, this work presents the first decision procedure and automated solver for decision problems in a separation logic fragment allowing both pointer arithmetic and memory blocks inside inductive definitions. 

\mypar{Organization.}
A motivating example and an overview of the decision procedures are provided in Section~\ref{sec:over-hls}.
The logic {\slah} is defined in Section~\ref{sec:logic-hls}.
Then the decision procedures for the satisfiability and entailment problems are presented in Sections~\ref{sec:sat-hls} \resp~\ref{sec:ent}.
%
%
The implementation details and the experimental evaluation
	are reported in Section~\ref{sec:exp-hls}.



\section{Motivating example and overview}
\label{sec:over-hls}

This section illustrates the use of the logic $\slah$ for the specification 
	of programs manipulating heap lists and
gives an overview of the ingredients used by the decision procedures
	we propose.

Figure~\ref{fig:search} presents the motivating example.
The function \lst{search} scans 
	a heap list between the addresses \lst{hbeg} (included) 
						and \lst{hend} (excluded)
	to find a chunk of size greater than the one given as parameter. 
A heap list is a block of memory divided into \emph{chunks},
such that each chunk stores its size at its start address.
For readability, we consider that the size counts 
the number of machine integers (and not bytes).
Scanning the list requires pointer arithmetics to
compute the start address of the next chunk.
After the size, the chunk may include additional information 
about the chunk in so-called  \emph{chunk's header}.
For readability, we consider 
that the header contains only the size information.
The data carried by the chunk, called \emph{chunk's body},
starts after the chunk header
and ends before the header of the next chunk.


\begin{figure}[t]
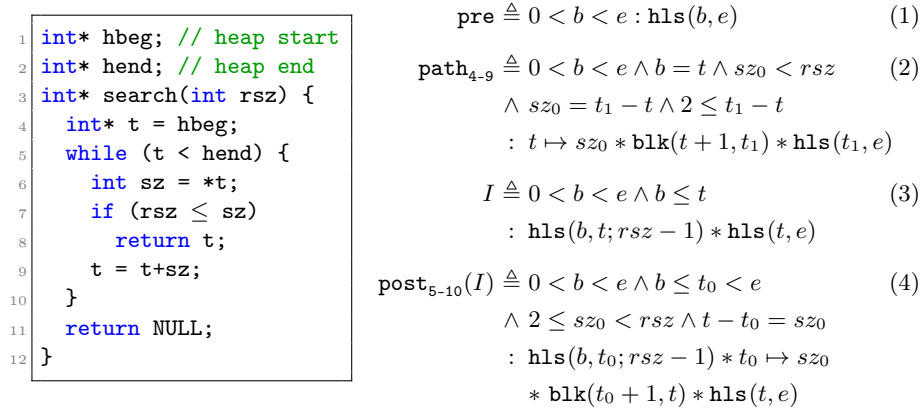

{\small
\centering
\begin{minipage}[t]{0.33\textwidth}
\vspace{6mm}
\lstset{language=C,numbers=left,stepnumber=1,fontadjust=true}
\begin{lstlisting}
int* hbeg; // heap start
int* hend; // heap end
int* search(int rsz) {
  int* t = hbeg;
  while (t < hend) {
    int sz = *t;
    if (rsz <= sz)
      return t;
    t = t+sz;
  }
  return NULL;
}
\end{lstlisting}
\end{minipage}
\begin{minipage}[t]{0.62\textwidth}
\vspace{0pt}
\begin{eqnarray}\label{eq:over-sat}
\texttt{pre} \label{eq:over-pre}
        & \triangleq & 0 < b < e : \hls{}(b,e) 
\\[2mm]
\texttt{path}_{\texttt{4-9}} \label{eq:over-path-once}
	    & \triangleq & 0 < b < e \land b = t \land sz_0 < rsz
\\
	    & \land & sz_0 = t_1 - t \land 2 \le t_1 - t  
\nonumber \\
        & :     & t \pto sz_0 \sepc \blk(t+1, t_1) \sepc \hls{}(t_1,e)
\nonumber \\[2mm]
I \label{eq:over-inv}
	    & \triangleq & 0 < b < e \land b \le t
	    \\
	    \nonumber
	    & : & \hls{}(b,t; rsz-1) \sepc \hls{}(t,e)
\\[2mm]
\texttt{post}_{\texttt{5-10}}(I) \label{eq:over-inv-once}
        & \triangleq &  0 < b < e \land b\le t_0 < e  \\
		& \land & 2 \le sz_0  < rsz \land t - t_0=sz_0 
\nonumber \\
		& :     & \hls{}(b,t_0;rsz-1) \sepc t_0 \pto sz_0 
\nonumber \\
		&       & \sepc\ \blk(t_0+1,t) \sepc \hls{}(t,e)
\nonumber
\end{eqnarray}
\end{minipage}
\vspace{-2eX}
\caption{Searching a chunk in a heap list 
and some of its specifications in \slah}
\label{fig:search}
}
\vspace{-4mm}
\end{figure}

The right part of Figure~\ref{fig:search} provides 
several formulas in the logic \slah\ defined in Section~\ref{sec:logic-hls}.
The formula \texttt{pre} specifies the precondition of \lst{search}, 
i.e., there is a heap list from \lst{hbeg} to \lst{hend}
represented by the logic variables $b$ resp. $e$.
The pure part of \texttt{pre} (left side of ':') is a linear constraint
on the addresses. The spatial part of \texttt{pre} (right side of ':') employs the
predicate $\hls{}$ defined inductively by the last two rules below
and for which we define the following shorthand:
\begin{align}
\label{eq:hlsv-infty}
\hls{}(x, y) & \equiv \hls{}(x,y;\infty) \mbox{ i.e., no upper bound on chunks' size, where}
\\
\label{eq:hlsv-emp}
\hls{}(x, y; v) & \Leftarrow x=y : \emp  \\
\label{eq:hlsv-rec}
\hls{}(x, y; v) & \Leftarrow \exists z \cdot 2 \le z - x \le v : 
	x \pto z-x \sepc \blk(x+1,z) \sepc \hls{}(z,y; v) 
\end{align}
The inductive definition states that a heap list from $x$ to $y$ is 
	either an empty heap if $x$ and $y$ are aliased (Eq.~(\ref{eq:hlsv-emp})),
    or a \emph{chunk} starting at address $x$ and ending at
    address $z$ followed by a heap list from $z$ to $y$.
    The chunk stores its size $z-x$ in the header
		(atom $x\pto z-x$).
	The chunk's body is specified by a memory block atom, $\blk(x+1,z)$, starting
		after the header and ending before the next chunk.
The parameter $v$ is an upper bound on the size of all chunks in the list.

The formula $\texttt{path}_{\texttt{4-9}}$ is generated by 
	the symbolic execution of \lst{search} 
	starting from \texttt{pre} and 
	executing the statements from line 4 to 9. 
It's satisfiability means that the line 9 is reachable from 
a state satisfying \texttt{pre}.
  
The decision procedure for the satisfiability of \slah\
	in Section~\ref{sec:sat-hls} is based
on the translation of a \slah\ formula $\varphi$ into 
an equi-satisfiable Presburger arithmetic (\PbA) formula $\varphi^P$. 
The delicate point 
with respect to the previous work, e.g., \cite{DBLP:journals/corr/abs-1802-05935},  
is to compute a summary in \PbA\ for the $\hls{}$ atoms.
The summary computed for the atom $\hls{}(x,y;v)$ 
when the heap it denotes is not empty
is
$(v=2 \land \exists k\cdot k > 0 \land 2k = y - x) \lor (2 < v \land 2 < y-x)$,
i.e., either all chunks have size 2 and the heap-list has an even size or
$v$ and the size of the heap-list are strictly greater than 2.
For the empty case, the summary is trivially $x=y$.
The other spatial atoms $a$ (e.g., $x \pto v$ and $\blk(x,y)$) 
are summarized by constraints on 
their start address denoted by $\atomhead(a)$ (e.g., $x$) and 
their end address denoted by $\atomtail(a)$ (e.g., $x+1$ resp. $y$).
For the points-to atom, this constraint is true, 
but for the $\blk(x,y)$ atom, the constraint is 
$x = \atomhead(a) < \atomtail(a) = y$.
Therefore, the spatial part of $\texttt{path}_{\texttt{4-9}}$ is translated 
into the \PbA\ formula $\texttt{pb}^\Sigma_{\texttt{4-9}}$:
\begin{align*}
\underbrace{t+1 < t_1}_{\blk{}(t+1,t_1)} & \land 
  \underbrace{(t_1 = e \lor 2 < e-t_1)}_{\hls{}(t_1,e)}.
\end{align*}
Alone, $\texttt{pb}^\Sigma_{\texttt{4-9}}$ does not capture the semantics of 
	the separating conjunction in the spatial part.
For that, we add a set of constraints 
	expressing the disjointness of memory blocks occupied by the spatial atoms.
For our example, these constraints are 
	$\texttt{pb}^\sepc_{\texttt{4-9}} \triangleq t_1 < e \le t ~\lor~ t+1 < t_1 \le e$.
By conjoining the pure part of $\texttt{path}_{\texttt{4-9}}$ with 
formulas $\texttt{pb}^\Sigma_{\texttt{4-9}}$ and 
         $\texttt{pb}^\sepc_{\texttt{4-9}}$,
we obtain an equi-satisfiable existentially quantified {\PbA} formula
 whose satisfiability is a NP-complete problem.

\smallskip
\noindent
The {\PbA} abstraction is also used to decide the validity of entailments 
in \slah\ in combination with a matching procedure between spatial parts.
To illustrate this decision procedure presented in Section~\ref{sec:ent},
we consider the verification conditions generated by the proof of 
 the invariant $I$ from Equation~(\ref{eq:over-inv}) for the \lst{search}'s loop. 
It states that \lst{t} splits the heap list in two parts. 
%
%
To illustrate a non-trivial case of the matching procedure used
	in the decision procedure for entailment,
	we consider the verification condition (VC)
	for the inductiveness of $I$.
	The antecedent of the VC is the formula 
	$\texttt{post}_{\texttt{5-10}}(I)$ in Figure~\ref{fig:search},
	obtained by symbolically executing the path including the statements
	at lines 5--7 and 9 starting from $I$.
The \PbA\ abstraction of $\texttt{post}_{\texttt{5-10}}(I)$ is satisfiable
and entails the following ordering constraint on the terms used by the spatial atoms:
$0 < b \le t_0 < t_0+1 < t \le e$.
The spatial atoms used in the antecedent and consequent 
are ordered using the order given by this contraint as follows:
\[\begin{array}{lcc}
\textrm{antecedent:} & \hls{}(b,t_0;rsz-1) \sepc\ t_0 \pto sz_0 \sepc\ \blk(t_0+1,t) 
				     & \sepc\ \hls{}(t,e) \\
\textrm{consequent:} & \hls{}(b,t;rsz-1) & \sepc\ \hls{}(t,e)
\end{array}\]
The matching procedure starts by searching 
	a prefix of the sequence of atoms in the antecedent
	that matches the first atom in the consequent, $\hls{}(b,t;rsz-1)$, such that
	the start and end addresses of the sequence are respectively $b$ and $t$.
The sequence found is $\hls{}(b,t_0;rsz-1) \sepc t_0 \pto sz_0 \sepc \blk(t_0+1,t)$
	which also satisfies the condition (encoded in \PbA)
		that it defines a contiguous memory block between $b$ and $t$.
The algorithm continues by trying to prove the matching found
	using a composition lemma $\hls{}(b,t_0;rsz-1) \sepc \hls{}(t_0,t;rsz-1) \models \hls{}(b,t;rsz-1)$
	and the unfolding of the atom $\hls{}(t_0,t;rsz-1)$.
The \PbA\ abstraction of the antecedent is used to ensure that 
    $sz_0$ is the size of the heap list starting at $t_0$, i.e., $sz_0=t-t_0$
    and the constraint $2 \le sz_0 \le rsz-1$ is satisfied.
	For this ordering of terms (i.e., $0 < b \le t_0 < t_0+1 < t \le e$), 
	the algorithm is able to prove the matching. Since this ordering is the only one compatible with the {\PbA} abstraction of the antecedent, we conclude that the entailment is valid.

\section{\slah, a separation logic fragment for heap-list}
\label{sec:logic-hls}

This section defines the syntax and the semantics of 
the logic \slah, which extends array separation logic (ASL)~\cite{BrotherstonGK17}
with the $\hls{}$ predicate.
Notice that ASL borrows the structure of formulas from 
	the symbolic heap fragment of SL~\cite{BerdineCO04}
	and introduces a new spatial atom for memory blocks.

\begin{definition}[\slah\ Syntax]\label{def:syn-slah}
Let $\cV$ denote an infinite set of address variables ranging over $\NN$, the set of natural numbers. 
The syntax of terms $t$, pure formulas $\Pi$, spatial formulas $\Sigma$
and symbolic heaps $\varphi$ is given by the following grammar: 
\[
\begin{array}{l l l r}
t & ::= x \mid n \mid t+t &\ \ \ & \mbox{terms}
\\ 
\Pi & ::= \top \mid \bot \mid t=t \mid t \ne t \mid t \leq t \mid t < t \mid \Pi \land \Pi &\ \ \  & \mbox{pure formulas}
\\ 
\Sigma & ::= \emp \mid t \pto t \mid \blk(t, t) \mid \hls{}(t, t; t^\infty) \mid \Sigma \sepc \Sigma &\ \ \  & \mbox{spatial formulas}
\\ 
\varphi & ::= \exists \vec{z}\cdot \Pi : \Sigma &\ \ \  & \mbox{formulas}
\end{array}
\]
where $x$ and $\vec{z}$ are variables resp. set of variables from $\cV$, 
$n \in \NN$, $t^\infty$ is either a term or $\infty$,
$\hls{}$ is the predicate defined inductively by 
the rules in Equations~(\ref{eq:hlsv-emp}) and (\ref{eq:hlsv-rec}),
where $v$ is a variable interpreted over $\NN\cup\{\infty\}$. 
An atom $\hls{}(x,y;\infty)$ is also written $\hls{}(x,y)$;
Whenever one of $\Pi$ or $\Sigma$ is empty, we omit the colon. 
We write $\fv(\varphi)$ for the set of free variables occurring in $\varphi$. 
If $\varphi = \exists\vec{z}\cdot\Pi:\Sigma$, 
	we write $\qf(\varphi)$ for $\Pi:\Sigma$, 
	the quantifier-free part of $\varphi$.
We define $\atomhead(a)$ and $\atomtail(a)$ where $a$ is a spatial atom as follows:
\begin{itemize}
\item if $a \equiv t_1 \pto t_2$
          then $\atomhead(a)\triangleq t_1$, $\atomtail(a)\triangleq t_1+1$,
\item if $a\equiv \blk(t_1, t_2)$
          then $\atomhead(a)\triangleq t_1$ and $\atomtail(a)\triangleq t_2$,
\item if $a\equiv \hls{}(t_1, t_2; t_3)$
      then $\atomhead(a) \triangleq t_1$ and $\atomtail(a) \triangleq t_2$.
\end{itemize}
\end{definition}

For $n, n' \in \NN$ such that $n \le n'$, we use $[n, n']$ to denote the set $\{n, \cdots, n'\}$. Moreover, we use $[n]$ as an abbreviation of $[1,n]$.
We interpret the formulas in the  classic model of separation logic
built from a stack $s$ and a heap $h$.
The stack is a function $s:\cV \tfun \NN$.
It is extended to terms, $s(t)$, to denote 
the interpretation of terms in the given stack;
$s(t)$ is defined by structural induction on terms:
$s(n) = n$, and
$s(t + t') = s(t) + s(t')$.
We denote $s[x\gets n]$ for a stack defined as $s$ except for the
interpretation of $x$ which is $n$.
Notice that $\infty$ is used only to give up the upper bound 
	on the value of the chunk size in the definition of $\hls{}$ 
	(see Equations~(\ref{eq:hlsv-emp})--(\ref{eq:hlsv-infty})).
The heap $h$ is a partial function $\NN \pfun \NN$. 
We denote by $\wpos{h}$ the domain of a heap $h$. 
We use $h_1 \uplus h_2$ to denote the \emph{disjoint} union of $h_1$ and $h_2$, that is, $\wpos{h_1} \cap \wpos{h_2} =\emptyset$, and for $i \in \{1,2\}$, we have $(h_1\uplus h_2)(n) = h_i(n)$ if $n \in \wpos{h_i}$. 

\begin{definition}[\slah\ Semantics]
The satisfaction relation $s,h\models \varphi$, 
	where $s$ is a stack, $h$ a heap, and $\varphi$ a \slah\ formula,
	is defined by: 
\begin{itemize}
\item $s,h \models  \top$ always and never $s,h \models  \bot$,
\item $s,h \models  t_1 \sim t_2$ iff   
	$s(t_1) \sim s(t_2)$, where $\sim\in\{=,\ne,\le,<\}$,
\item
$s,h \models  \Pi_1 \land \Pi_2$  iff  
	$s,h \models \Pi_1$ and  $s,h\models \Pi_2$,
\item
$s,h \models  \emp$  iff $\wpos{h} = \emptyset$,
\item 
$s,h \models  t_1 \pto t_2$ iff $\exists n\in\NN$ s.t.
	$s(t_1)=n$, $\wpos{h}=\{n\}$, and $h(n)=s(t_2)$,
%
\item $s,h \models  \blk(t_1, t_2)$ iff  
	$\exists n,n'\in\NN$ s.t.  $s(t_1)=n$, $s(t_2)=n'$, $n < n'$, and
	  $\wpos{h}=[n,n'-1]$,
	 
\item $s,h \models  \hls{}(t_1, t_2;t_3)$ iff  
	$\exists k \in\NN$ s.t. $s,h \models \hls{k}(t_1, t_2;t_3)$,
	
\item $s,h \models \hls{0}(t_1, t_2;t^\infty)$ iff  
	$s,h \models t_1 = t_2 : \emp$,

\item 
$s,h \models \hls{\ell+1}(t_1, t_2;t^\infty)$ iff   
	$s,h \models 
		\exists z\cdot  2 \le z-t_1 \land \Pi' : t_1\pto z-t_1\ \sepc 
		 \blk(t_1+1,z) \sepc \hls{\ell}(z,t_2;t^\infty),
	$
	 where if  $t^\infty\equiv\infty$, then  $\Pi'\equiv \top$, otherwise,  
	 $\Pi' \equiv z-t_1 \le t^\infty$,

\item 
$s,h \models  \Sigma_1 \sepc \Sigma_2  \mbox{ iff } 
	\exists h_1,h_2\mbox{ s.t. } h=h_1\uplus h_2, 
	s,h_1\models\Sigma_1 \mbox{ and } s,h_2\models\Sigma_2$,

\item 
$s,h \models  \exists\vec{z}\cdot\Pi:\Sigma$ iff $\exists \vec{n}\in\NN^{|\vec{z}|}$ s.t.  
	$s[\vec{z}\gets\vec{n}],h\models\Pi$ and  $s[\vec{z}\gets\vec{n}],h\models\Sigma$.
\end{itemize}
\end{definition}

We write $A \models B$ for $A$ and $B$ (sub-)formula in \slah\ for
$A$ entails $B$, i.e., 
that for any model $(s,h)$ such that $s,h\models A$ then $s,h\models B$.

Notice that the semantics of $\blk(x,y)$ differs from the one given in~\cite{BrotherstonGK17} for $\mathtt{array}(x,y)$
because we consider that the location $y$ is the first after the last location in the memory block, as proposed in~\cite{CalcagnoDOHY06}.
Intuitively, an atom $\hls{}(x,y;v)$ with $v$ a variable
defines a heap lists where all chunks have sizes between $2$ and the value of $v$.
Notice that if $v < 2$ then the atom $\hls{}(x,y;v)$ has a model iff $x=y$.
With this semantics, the $\blk$ and $\hls{}$ predicates are compositional predicates~\cite{EneaSW15} and therefore they satisfy the following composition lemmas:
\begin{align}
\blk(x,y)\sepc\blk(y,z) \models & ~\blk(x,z) \label{lemma:blk}
\\
\hls{}(x,y;v)\sepc\hls{}(y,z;v) \models & ~\hls{}(x,z;v)  
\label{lemma:hls}
\end{align}




\section{Satisfiability problem of {\slah}}
\label{sec:sat-hls}

The satisfiability problem for an \slah\ formula $\varphi$ is to decide
whether there is a stack $s$ and a heap $h$ such that $s,h \models \varphi$.
In this section, we propose a decision procedure for the satisfiability problem, 
thus showing that the satisfiability problem is NP-complete.

\begin{theorem}\label{thm-sat}
The satisfiability problem of {\slah} is NP-complete.
\end{theorem}
The NP lower bound follows from that of ASL in \cite{BrotherstonGK17}. The NP upper bound is achieved by encoding the satisfiability problem of {\slah} as that of an existentially quantified Presburger arithmetic formula. The rest of this section is devoted to the proof of the NP upper bound.

Presburger arithmetic (\PbA) is the first-order theory with equality
of the signature $\langle 0,1,+, <, (\equiv_n)_{n \in \NN \setminus \{0\}}\rangle$ interpreted over the
domain of naturals $\NN$ with `$+$' interpreted as the addition,  
`$<$' interpreted as the order relation, and $\equiv_n$ interpreted as the congruence relation modulo $n$.\footnote{%
Although `$<$' may be encoded using existential quantification in \PbA\ over naturals,
we prefer to keep it in the signature of \PbA\ to obtain quantifier free formulas.%
}
\PbA\ is a useful tool for showing complexity classes because its 
satisfiability problem belongs to various complexity classes depending
on the number of quantifier alternations~\cite{Haase2018ASG}.
In this paper, we consider quantifier-free \PbA\ formulas (abbreviated as \qfpa) and the $\mathsf{\Sigma}_1$-fragment of \PbA\ (abbreviated as \EPbA), which
contains existentially quantified Presburger arithmetic formulas. 
We recall that the satisfiability problem of {\qfpa} and {\EPbA} is NP-complete.

%

We basically follow the same idea as ASL to build 
a {\qfpa} abstraction of a \slah\ formula $\varphi$, denoted by $\abs(\varphi)$,
that encodes its satisfiability: 
\begin{compactitem}
\item At first, points-to atoms $t_1 \pto t_2$ are transformed into $\blk(t_1, t_1+1)$. 
\item Then, the block atoms $\blk(t_1, t_2)$ are encoded by the constraint $t_1 < t_2$.
\item The predicate atoms $\hls{}(t_1, t_2; t_3)$, absent in ASL, are encoded by a formula in {\qfpa}, $t_1 = t_2 \vee (t_1 < t_2 \wedge \abs^+(\hls{}(t_1, t_2; t_3)))$. 
\item Lastly, the separating conjunction is encoded by an {\qfpa} formula constraining the address terms of spatial atoms. 
\end{compactitem}
The Appendix~\ref{app:sat-hls} provides more details.
The crux of this encoding and 
its originality with respect to the ones proposed for ASL in~\cite{BrotherstonGK17}
is the computation of $\abs^+(\hls{}(t_1, t_2; t_3))$, 
which are the least-fixed-point summaries in {\qfpa} for $\hls{}(t_1, t_2; t_3)$.
In the sequel, we show how to compute them.





%
%

\smallskip
Intuitively, the abstraction of the predicate atoms $\hls{}(t_1, t_2; t_3)$
	shall summarize the relation between $t_1$, $t_2$ and $t_3$ 
	for all $k \ge 1$ unfoldings of the predicate atom.
From the fact that the pure constraint in the inductive rule of $\hls{}$ is $2 \le x' - x \le v$, it is easy to observe that 
for each $k \ge 1$, $\hls{k}(t_1, t_2; t_3)$ can be encoded by $2 k \le t_2-t_1 \le k t_3$. It follows that $\hls{}(t_1, t_2; t_3)$ can be encoded by $\exists k.\ k \ge 1 \wedge 2k \le t_2 - t_1 \le k t_3$.
If $t_3 \equiv\infty$, then $\exists k.\ k \ge 1 \wedge 2 k \le t_2-t_1 \le k t_3$ is equivalent to  $\exists k.\ k \ge 1 \wedge 2k \le t_2 - t_1 \equiv 2 \le t_2 - t_1$, thus a {\qfpa} formula.
Otherwise, $2 k \le t_2-t_1 \le k t_3$ is a non-linear formula since $k t_3$ is a non-linear term if $t_3$ contains variables. 
The following lemma states that 
$\exists k.\ k \ge 1 \wedge 2 k \le t_2 - t_1 \le k t_3$ can actually be turned into an equivalent {\qfpa} formula.

\begin{lemma}[Summary of $\hls{}$ atoms]\label{lem-hls}
Let $ \hls{}(x, y; z)$ be an atom in \slah\
representing a non-empty heap, where $x, y, z$ are three distinct variables in $\cV$.
We can construct in polynomial time an {\qfpa} formula $\abs^+(\hls{}(x,y; z))$
which summarizes $\hls{}(x, y; z)$, namely we have 
for every stack $s$, $s \models \abs^+(\hls{}(x,y; z))$ iff 
there exists a heap $h$ such that $s, h \models \hls{}(x, y, z)$. 
\end{lemma}

Since the satisfiability problem of {\qfpa} is NP-complete,
the satisfiability problem of \slah\ is in NP. 
%
%
The correctness of $\abs(\varphi)$ is guaranteed by the following result.

\begin{proposition}\label{prop-sat-correct}
A \slah\ formula $\varphi$ is satisfiable iff $\abs(\varphi)$ is satisfiable.
\end{proposition}

From now on, we shall assume that {\bf $\abs(\varphi)$ is a \qfpa\ formula}. This enables using the off-the-shelf SMT solvers, e.g. Z3, to solve the satisfiability problem of \slah\ formulas.





\section{Entailment problem of {\slah}}
\label{sec:ent}

We consider the following entailment problem:
Given the symbolic heaps $\varphi$ and $\psi$ in \slah\ such that 
	$\psi$ is quantifier free and 
	$\fv(\psi) \subseteq \fv(\varphi)$,
	decide if $\varphi \models \psi$.
	Notice that the existential quantifiers in $\varphi$, if there is any, can be dropped.


	
Our goal in this section is to show that the entailment problem is decidable, as stated in the following theorem.

\begin{theorem}\label{thm-entail}
The entailment problem of {\slah} is coNP-complete. 
\end{theorem}
The coNP lower bound follows from the fact that the entailment problem of quantifier-free ASL formulas is also coNP-complete~\cite{BrotherstonGK17}. 
The remainder of this section is devoted to the proof of the coNP upper bound. 
%

In a nutshell, we show in Section~\ref{ssec:order}
that the entailment problem $\varphi \models \psi$ 
can be decomposed into a finite number of \emph{ordered entailment problems} 
$\varphi' \models_{\preceq} \psi'$ where
all the terms used as start and end addresses of spatial atoms in $\varphi'$ and $\psi'$
 are ordered by a preorder $\preceq$.
Then, we propose a decision procedure to solve ordered entailment problems.
In Section~\ref{ssec:ent-1}, we consider the special case    
    where the consequent $\psi'$ has a unique spatial atom; 
    this part reveals a delicate point which appears when the consequent 
    and the antecedent are $\hls{}$ atoms because of the constraint on the chunk sizes.
The general case is dealt with in Section~\ref{ssec:ent-all}; 
    the procedure calls the special case for the first atom of the consequent with 
    all the compatible prefixes of the antecedent, and 
	it does a recursive call for the remainders of the consequent and the antecedent. Note that to find all the compatible prefixes of the antecedent, some spatial atoms in the antecedent might be split into several ones. We derive the coNP upper bound from the aforementioned decision procedure as follows:  
\begin{enumerate}
\item The entailment problem is reduced to at most exponentially many 
\emph{ordered entailment problems} since there are exponentially many total preorders.
\item Each ordered entailment problem can be reduced further to exponentially many \emph{special ordered entailment problems} where there is one spatial atom in the consequent.
\item The original entailment problem is invalid iff there is an invalid special ordered entailment problem instance.
\item The special ordered entailment problem is in coNP.
\end{enumerate}

In the sequel, we assume that $\abs(\varphi)$ is satisfiable and $\abs(\varphi) \models \abs(\psi)$. Otherwise, the entailment is trivially unsatisfiable.

\subsection{Decomposition into ordered entailments}
\label{ssec:order}

Given the entailment problem $\varphi \models \psi$,
we denote by $\addr(\varphi)$ (and $\addr(\psi)$) 
	the set of terms used as start and end addresses of spatial atoms 
	in $\varphi$ (resp. $\psi$).
%
Recall that a \emph{preorder} $\preceq$ over a set $A$ 
	is a reflexive and transitive relation on $A$. 
	The preorder $\preceq$ on $A$ is \emph{total} if for every $a, b \in A$, 
	either $a \preceq b$ or $b \preceq a$. 
	For $a,b\in A$, we denote by $a \simeq b$ the fact that 
		$a \preceq b$ and $b \preceq a$,
	and we use $a \prec b$ for  $a \preceq b$ but not $b \preceq a$.

\begin{definition}[Total preorder compatible with $\abs(\varphi)$]
Let  $\preceq$ be a total preorder over $\addr(\varphi) \cup \addr(\psi)$. Then $\preceq$  is said to be \emph{compatible with}  $\varphi$ if
$C_\preceq \wedge \abs(\varphi)$ is satisfiable, where
\begin{align}\label{eq:start-end}
C_\preceq & \triangleq &
	\bigwedge_{t_1, t_2 \in \addr(\varphi) \cup \addr(\psi), t_1  \simeq  t_2} t_1 = t_2 \wedge \bigwedge_{t_1, t_2 \in \addr(\varphi) \cup \addr(\psi), t_1  \prec  t_2} t_1 < t_2.
\end{align}
\end{definition}
\begin{example}
Let $\varphi \equiv  \blk(x_1, x_2) \sepc \hls{}(x_2, x_3; y)$ and $\psi \equiv \blk(x_1, x_3)$. Then $\addr(\varphi) \cup \addr(\psi) = \{x_1, x_2, x_3\}$. From $\abs(\varphi) \models x_1 < x_2 \wedge x_2 \le x_3$, there are two total preorders compatible with $\varphi$, namely, $x_1 \prec_1 x_2 \prec_1 x_3$ and $x_1 \prec_2 x_2 \simeq_2 x_3$.
\end{example}

\begin{definition}[$\varphi \models_\preceq \psi$]
Let $\preceq$ be a total preorder over $\addr(\varphi) \cup \addr(\psi)$ that is compatible with $\varphi$. Then we say $\varphi \models_\preceq \psi$ if $C_\preceq \wedge \Pi: \Sigma \models \Pi': \Sigma'$. 
\end{definition}

\begin{lemma}\label{lem:split-prec}
$\varphi \models \psi$ iff 
for every total preorder $\preceq$ over $\addr(\varphi) \cup \addr(\psi)$ that is compatible with $\varphi$,
we have $\varphi \models_\preceq \psi$.
\end{lemma}
The proof of the above lemma is immediate.
There may be exponentially many total preorders over $\addr(\varphi) \cup \addr(\psi)$ that are compatible with $\varphi$ in the worst case.


The procedure to decide $\varphi \models_\preceq \psi$ is presented in the rest of this section.
We assume that $\varphi \equiv \Pi: a_1 \sepc \cdots \sepc a_m$ and $\psi \equiv \Pi': b_1 \sepc \cdots \sepc b_n$ such that
\begin{align}
C_\preceq \wedge \abs(\varphi) \mbox{ is satisfiable and }
C_\preceq \wedge \abs(\varphi) \models \abs(\psi).
\label{eq:order}
\end{align}
We consider that the atoms $\hls{}(t_1, t_2; t_3)$ in $\varphi$ or $\psi$ such that $C_\preceq \models t_1 = t_2$ are removed because they correspond to an empty heap.
Moreover, after a renaming, 
we assume that the spatial atoms are sorted in each formula, namely, 
the following two \PbA\ entailments hold:
\begin{align}
C_\preceq \models & \bigwedge \limits_{i \in [1,m]} \atomhead(a_i) < \atomtail(a_i) \wedge \bigwedge \limits_{1 \le i < m}  \atomtail(a_i) \le  \atomhead(a_{i+1}),
\label{eq:order-ante} \\
C_\preceq \models & \bigwedge \limits_{i \in [1,n]} \atomhead(b_i) < \atomtail(b_i) \wedge \bigwedge \limits_{1 \le i < n}  \atomtail(b_i) \le  \atomhead(b_{i+1}).
\label{eq:order-conseq}
\end{align}

Section~\ref{ssec:ent-1} considers the special case of a consequent $\psi$ having only one spatial atom. Section~\ref{ssec:ent-all} considers the general case. 

\subsection{Consequent with one spatial atom}
\label{ssec:ent-1}

Consider the ordered entailment 
$\varphi \models_\preceq \psi$, where 
$\varphi \equiv \Pi: a_1 \sepc \cdots \sepc a_m$, 
$\psi \equiv \Pi': b_1$
and the constraints (\ref{eq:order})--(\ref{eq:order-conseq}) are satisfied.
From~(\ref{eq:order}) and the definition of $\abs$, we have that 
$C_\preceq \land \abs(\varphi)$ implies $\Pi'$, 
so we simplify this entailment to deciding:
\begin{align*}
C_\preceq \wedge \Pi: a_1 \sepc \cdots \sepc a_m \models_\preceq b_1,
\end{align*}
where, the atoms $a_i$ ($i\in [m]$) and $b_1$ represent \emph{non-empty} heaps,
and the start and end addresses of atoms $a_i$ as well as those of $b_1$ are totally ordered by $C_\preceq$.
  
Because $b_1$ defines a continuous memory region, 
the procedure checks the following necessary condition in \PbA: 
{\small
\begin{align*}
C_\preceq \models 
\atomhead(a_1) = \atomhead(b_1) \wedge \atomtail(a_m) = \atomtail(b_1) \wedge \bigwedge \limits_{1 \le i < m} \atomtail(a_i) = \atomhead(a_{i+1}).
\end{align*}
}
\noindent
Then, the procedure does a case analysis on the form of $b_1$.
If $b_1 \equiv t_1 \pto t_2$ 
	then $\varphi  \models_\preceq \psi$ holds iff $m = 1$ and $a_1 = t'_1 \pto t'_2$.
If $b_1 \equiv \blk(t_1, t_2)$ 
	then $\varphi  \models_\preceq \psi$ holds.
For the last case, 
$b_1 \equiv \hls{}(t_1, t_2; t_3)$, we distinguish between $m=1$ or not.

\subsubsection{One atom in the antecedent:} A case analysis on the form of $a_1$ follows.

\myfpar{$a_1 \equiv t'_1 \pto t'_2$}
	Then $\varphi  \models_\preceq \psi$ does not hold, since a nonempty heap 
	modeling $b_1$ has to contain at least two memory cells.

\myfpar{$a_1 \equiv \blk(t'_1, t'_2)$}
	Then the entailment $\varphi  \models_\preceq \psi$ does not hold
	because a memory block of size $t'_2-t'_1$
	where the first memory cell stores the value $1$ 
	satisfies $\blk(t'_1, t'_2)$
	but does not satisfy $b_1 \equiv \hls{}(t_1, t_2; t_3)$
	where, by the inductive rule of $\hls{}$, 
	$t_1 \pto z - t_1$ and $2 \le z - t_1 \le t_3$.

\myfpar{$a_1 \equiv \hls{}(t'_1, t'_2; t'_3)$}
	Then the entailment problem seems easy to solve. 
	One may conjecture that 
	$C_\preceq \wedge \Pi: \hls{} (t'_1, t'_2; t'_3) \models \hls{}(t_1, t_2; t_3)$ 
	iff $C_\preceq \wedge \abs(\varphi) \models t'_3 \le t_3$, 
	which is \emph{not} the case as a matter of fact, 
	as illustrated by the following example. 
(Recall that, from (\ref{eq:start-end}) we have that $C_\preceq \wedge \abs(\varphi) \models t'_1 = t_1 \wedge t'_2 = t_2$.)

\begin{example}
Consider $x < y \wedge y -x  =  4: \hls{}(x, y; 3) \models \hls{}(x, y; 2)$. The entailment is valid, while we have $3 > 2$.
The reason behind this seemly counterintuitive fact is that when we unfold $ \hls{}(x, y; 3)$ to meet the constraint $y - x = 4$, it is impossible to have a memory chunk of size $3$. (Actually every memory chunk is of size $2$ during the unfolding.) 
\end{example}

We are going to show how to tackle this issue in the sequel.

\begin{definition}[Unfolding scheme of a predicate atom and effective upper bound]
Let $\varphi \equiv \Pi: \hls{}(t'_1, t'_2; t'_3)$ be an {\slah} formula and $s: \cV \rightarrow \NN$ be a stack such that $s \models \abs(\varphi)$ and $s(t'_2) - s(t'_1) \ge 2$. 
An \emph{unfolding scheme} of $\varphi$ w.r.t. $s$ is a sequence of numbers
$(sz_1, \cdots, sz_\ell)$ such that $2 \le sz_i \le s(t'_3)$ for every $i \in [\ell]$ and $s(t'_2) = s(t'_1) + \sum_{i \in [\ell]} sz_i$. 
Moreover, $\max(sz_1, \cdots, sz_\ell)$ is called the \emph{chunk size upper bound} associated with the unfolding scheme.
The \emph{effective upper bound} of $\varphi$ w.r.t. $s$, denoted by $\eub_\varphi(s)$, is defined as the maximum chunk size upper bound associated with the unfolding schemes of $\varphi$ w.r.t. $s$.
\end{definition}

\begin{example}
Let $\varphi \equiv x < y: \hls{}(x, y; 3)$ and $s$ be a store such that $s(x)= 1$ and $s(y) = 7$. Then there are two unfolding schemes of $\varphi$ w.r.t. $s$, namely, $(2, 2, 2)$ and $(3,3)$, whose chunk size upper bounds are $2$ and $3$ respectively. Therefore, $\eub_\varphi(s)$, the effective upper bound of $\varphi$ w.r.t. $s$, is $3$.
\end{example}

The following lemma (proved in the appendix) states that 
the effective upper bounds of chunks in heap lists atoms of $\varphi$ 
with respect to stacks 
can be captured by a {\qfpa} formula.

\begin{lemma}\label{lem-eub}
For an {\slah} formula $\varphi \equiv \Pi: \hls{}(t'_1, t'_2; t'_3)$, a {\qfpa} formula $\xi_{eub,\varphi}(z)$ can be constructed in linear time such that for every store $s$ satisfying $s \models \abs(\varphi)$, we have $s[z \gets \eub_\varphi(s)] \models \xi_{eub,\varphi}(z)$ and $s[z \gets n] \not \models \xi_{eub,\varphi}(z)$ for all $n \neq \eub_\varphi(s)$.
\end{lemma}

The following lemma (proof in the appendix) provides the correct test used 
	for the case $a_1 \equiv \hls{}(t'_1, t'_2; t'_3)$.
\begin{lemma}\label{lem-hls-hls}
Let $\varphi \equiv \Pi: \hls{}(t'_1, t'_2; t'_3)$, $\psi \equiv \hls{}(t_1, t_2; t_3)$, and $\preceq$ be a total preorder over $\addr(\varphi) \cup \addr(\psi)$ such that $C_\preceq \models t'_1 < t'_2 \wedge t'_1 = t_1 \wedge t'_2 = t_2$. 
Then $\varphi \models_\preceq \psi$ iff 
$C_\preceq \wedge \abs(\varphi)  \models  \forall z.\ \xi_{eub, \varphi}(z) \rightarrow z \le t_3$.  
\end{lemma}
From Lemma~\ref{lem-hls-hls}, it follows that $\varphi \equiv \Pi: \hls{}(t'_1, t'_2; t'_3) \models_\preceq \hls{}(t_1, t_2; t_3)$ is invalid iff $C_\preceq \wedge \abs(\varphi)  \wedge  \exists z.\ \xi_{eub, \varphi}(z) \wedge \neg z \le t_3$ is satisfiable, which is an {\EPbA} formula. Therefore, this special case of the ordered entailment problem is in coNP.

\medskip
\noindent {\bf At least two atoms in the antecedent:} Recall that 
$\varphi\equiv C_\preceq \wedge \Pi: a_1 \sepc \cdots \sepc a_m$; 
a case analysis on the form of the first atom of the antecedent, $a_1$, follows. 

\myfpar{$a_1\equiv \blk(t_1', t_2')$} 
	Then $\varphi \models_\preceq \hls{}(t_1, t_2; t_3)$ does not hold (see case $m=1$).

\myfpar{$a_1\equiv \hls{}(t_1', t_2'; t'_3)$}
	Then $\varphi \models_\preceq \hls{}(t_1, t_2; t_3)$ 
	 iff 
	 $\abs(\varphi): \hls{}(t'_1, t'_2; t'_3) \models_\preceq \hls{}(t_1, t'_2; t_3)$
	 and 
	 $\abs(\varphi): a_2 \sepc \cdots \sepc a_m \models_\preceq \hls{}(t'_2, t_2; t_3)$.

\myfpar{$a_1\equiv t'_1 \pto t'_2$}
	Then the analysis is more involved because we have to check that 
	$t'_2$ is indeed the size of the first chunk in $\hls{}(t_1, t_2; t_3)$
		(i.e., satisfies $2\leq t'_2\leq t_3$)
	and the address $t'_1+t'_2$, the end of the chunk starting at $t'_1=t_1$, 
        is the start of a heap list in the antecedent.
    The last condition leads to the following cases:
\begin{compactitem}
\item$t'_1+t'_2$ is the end of some $a_j$
	where $j \in [m]$ such that 
	$t'_1 + t'_2 = \atomtail(a_j) \wedge C_\preceq \wedge \abs(\varphi)$ is satisfiable.
	Then the following entailment shall hold: 
{\small
\begin{align*}
2 \leq t'_2\leq t_3 \land t'_1 + t'_2 = \atomtail(a_j) \land \abs(\varphi): 
   a_{j+1} \sepc \cdots \sepc a_m & \models_\preceq \hls{}(t'_1+t'_2, t_2; t_3).
\end{align*}
}
%
\item$t'_1+t'_2$ is inside a block atom $a_j$:
	where $j \in [m]$ such that $a_j\equiv \blk(t''_1, t''_2)$ and 
	$t''_1 < t'_1 + t'_2 < t''_2 \land \abs(\varphi)$ is satisfiable. 
	Then $\varphi \not\models_\preceq \psi$ because
	a block atom cannot match the head of a heap list in the consequent.

\item$t'_1+t'_2$ is inside a heap-list atom $a_j$ 
	where $j \in [m]$ such that $a_j\equiv \hls{}(t''_1, t''_2; t''_3)$, 
	$t''_1 < t'_1 + t'_2 < t''_2 \land \abs(\varphi)$ is satisfiable.
	Then the following ordered entailment
	stating that the suffix of the antecedent starting at $t_1'+t_2'$ 
	      matches the tail of the consequent,
	shall hold: 
\begin{align*}
2 \leq t'_2\leq t_3 \land t''_1 < t'_1 + t'_2 < t''_2 \land \abs(\varphi):~ & \\
	  \hls{}(t'_1 + t'_2, t''_2; t''_3) 
	  \sepc\ a_{j+1} \sepc \cdots \sepc a_m &
	  \models_\preceq \hls{}(t'_1+t'_2, t_2; t_3)
\end{align*}
%
     and the following formula, 
     expressing that $t_1'+t_2'$ is inside a block of a chunk in $a_j$,
     shall be unsatisfiable (since otherwise the remaining suffix of the antecedent will start by a block atom and cannot match a heap list):
\[\abs(
	 \begin{array}[t]{l} t''_1 \le  x' < t'_1 + t'_2 < x'' \le t''_2 \wedge 2 \le x'' - x' \le t''_3 \wedge C_\preceq \wedge \Pi : \\
	a_1 \sepc \cdots \sepc a_{j-1} \sepc {\sf Ufld}_{x', x''}(\hls{}(t''_1, t''_2; t''_3))\ \sepc a_{j+1} \sepc \cdots \sepc a_m)
	\end{array}
\]
      where $x', x''$ are two fresh variables and 
      the formula {\sf Ufld} specifies a splitting of $a_j$
      into a heap list from $t''_1$ to $x'$, a chunk starting at $x'$ and ending at $x''$,
      and a heap list starting at $x''$:
\begin{align*}
{\sf Ufld}_{x', x''}(\hls{}(t''_1, t''_2; t''_3)) \triangleq\  
           & \hls{}(t''_1, x'; t''_3) \sepc x' \pto x''- x' \sepc \\
           & \blk(x'+1, x'') \sepc \hls{}(x'', t''_2; t''_3),
\end{align*}
\end{compactitem}
Notice that all $j\in[m]$ shall be considered above; 
if one $j$ satisfying the premises does not lead to a valid conclusion
then the entailment is not valid.

\vspace{-2mm}


\subsection{Consequent with more spatial atoms}
\label{ssec:ent-all}

Using the arguments similar to the one given for the case $n=1$, 
we simplify $\varphi\models_\preceq \psi$ with
	$\varphi \equiv \Pi: a_1 \sepc \cdots \sepc a_m$ and 
	$\psi \equiv \Pi': b_1 \sepc \cdots \sepc b_n$ to:
\begin{align*}
C_\preceq \land \Pi: a_1 \sepc \cdots \sepc a_m \models_\preceq b_1\sepc \cdots \sepc b_n.
\end{align*}
\hide{
The entailment for the general case is solved by the procedure $\mathtt{checkSortedEntl}$ in Algorithm~\ref{algorithm:checkSortedEntl}, where the most technical part is the situation that $\abs(\varphi) \vDash \atomhead(a_k) < \atomtail(b_1)< \atomtail(a_k)$ for some $k$ and $a_k = \hls{}(t'_1, t'_2; t'_3)$.
}

For $n>1$, the decision procedure tries all the possible 
partitions of the sequence $a_1 \sepc \cdots \sepc a_m$ into 
a prefix $a_1 \sepc \cdots \sepc a_k'$ to be matched by $b_1$ and 
a suffix $a_k''\sepc \cdots \sepc a_m$ to be matched by $b_2 \sepc \cdots \sepc b_n$,
where $a_k'$ and $a_k''$ are obtained by splitting the atom $a_k$.
The partition process depends on the relative ordering of 
$\atomtail(b_1)$, $\atomhead(a_k)$ and $\atomtail(a_k)$. 
Formally, for every $k\in[m]$, the procedure considers all the following cases
for which it generates recursive calls to check the entailments:
\begin{align*}
C_\preceq \land \Pi_k: a_1 \sepc \cdots \sepc a_k' & \models_\preceq b_1  
\mbox{ and }C_\preceq \land \Pi_k: a_k'' \sepc \cdots \sepc a_m & \models_\preceq b_2 \sepc \cdots \sepc b_n,
\end{align*}
where $\Pi_k, a'_k, a''_k$ are defined as follows.
\begin{compactitem}
\item If $\atomtail(b_1) = \atomtail(a_k) \land C_\preceq \land \abs(\varphi)$ is satisfiable, 
	then $a_k'\triangleq  a_k$, $a_k''\triangleq  \emp$, 
	     $\Pi_k\triangleq \atomtail(b_1) = \atomtail(a_k) \land \abs(\varphi)$.
\item If $\atomhead(a_k) < \atomtail(b_1) < \atomtail(a_k) \land C_\preceq \land \abs(\varphi)$ is satisfiable, 
	  then a case analysis on the form of $a_k$ is done
	  to apply the suitable composition lemma:
\begin{compactitem}
\item If $a_k = \blk(t''_1, t''_2)$, then 
	$a_k'\triangleq\blk(t''_1, \atomtail(b_1))$,
	$a_k''\triangleq\blk(\atomtail(b_1), t''_2)$, and 
	$\Pi_k\triangleq t''_1 < \atomtail(b_1) < t''_2 \land \abs(\varphi)$.
\item If $a_k = \hls{}(t''_1, t''_2; t''_3)$, then we distinguish the following cases:
\begin{compactitem}
\item $\atomtail(b_1)$ starts a chunk in $a_k$, that is,
      $a_k' \triangleq \hls{}(t''_1, \atomtail(b_1); t''_3)$,
      $a_k''\triangleq \hls{}(\atomtail(b_1), t''_2; t''_3)$, and 
      $\Pi_k\triangleq t''_1 < \atomtail(b_1) < t''_2 \land \abs(\varphi)$;
      
\item $\atomtail(b_1)$ splits the body of a chunk in $a_k$ starting at some (fresh)
      address $x$; 
      depending on the position of the chunk in the list (the first, the last, in the middle, or the only chunk in the heap list), we obtain four cases. 
      This case splitting is due to the fact that the ordered entailment problems always assume non-empty $\hls{}$ atoms in the formulas.
      For example, if the chunk starting at $x$ is in the middle of the heap list, then 
      $a_k'\triangleq \hls{}(t''_1, x; t''_3) \sepc x \pto x' - x \sepc \blk(x+1, \atomtail(b_1))$,
      $a_k''\triangleq \blk(\atomtail(b_1), x') \sepc \hls{}(x', t''_2; t''_3)$, and
      $\Pi_k\triangleq t''_1 < x < \atomtail(b_1) < x' < t''_2 \wedge \abs(\varphi)$,
      where both $x$ and $x'$ are fresh variables.
\end{compactitem}
\end{compactitem}
\end{compactitem}

\section{Implementation and experiments}
\label{sec:exp-hls}

\mypar{Implementation.}
The decision procedures presented are implemented as an extension of the \cspen\ solver~\cite{GuCW16}, called \cspenp, available at~\cite{CompSpenSite}.
Let us briefly recall some information about {\cspen}. {\cspen} is written in C++ and includes several decision procedures for symbolic heap fragments including 
(i) inductive predicates that are compositional~\cite{EneaSW15} 
	(the predicate $\ls{}(x,y)$ for list segments is a simple example)
and (ii) integer data constraints.
It uses SMT solvers (e.g., \textsc{Z3}) 
for solving linear integer arithmetic constraints.  
\cspen\ ranked third among the eleven solvers in the general podium 
of the last edition of SL-COMP, the competition of separation-logic solvers~\cite{SLCOMPsite}.

\cspenp\ 
supports the new theory \slah. 
Internally, \cspenp\ parses the input file 
which shall include the definition of $\hls{}$ 
and the satisfiability or entailment queries in the SL-COMP format~\cite{SLCOMPsite}.
%
\begin{compactitem}
\item For satisfiability queries, 
	it constructs the \EPbA\ abstraction of the \slah\ formulas 
    as shown in Section~\ref{sec:sat-hls}, and queries an SMT solver 
    on its satisfiability.
%
\item For entailment queries $\varphi \models \psi$, 
	\cspenp\ has to enumerate all the total preorders over the set of
	the start and end addresses of spatial atoms in $\varphi$ and $\psi$, 
	as described in Section~\ref{sec:ent}. This is time-consuming and 
	a bottleneck for the performance. 
	We introduced some heuristics based on the preprocessing of the formula 
	to extract preorders between addresses or 
	to decompose the entailment on simpler ones (i.e., with less atoms).
	For instance, we check for every block atom $b$ in $\psi$, 
	whether there is a collection of spatial atoms, 
	say $a_i \sepc a_{i+1} \sepc \cdots \sepc a_j$ with $i < j$, 
	such that they are contiguous 
	(i.e., the ending address of $a_k$ is the starting address of $a_{k+1}$),
	$\atomhead(b) = \atomhead(a_i)$ and $\atomtail(b) = \atomtail(a_j)$. 
	If this is the case, then we generate the entailment query $\abs(\varphi): a_i \sepc a_{i+1} \sepc \cdots \sepc a_j \models b$ and remove all these spatial atoms from $\varphi$ and $\psi$, thus reducing the original entailment query to a smaller one, for which the number of addresses for the total preorder enumeration is decreased.

%
	
\end{compactitem}

\smallskip
\mypar{Benchmarks.}
We generated 190 benchmarks, 
available at~\cite{benchmark},
classified into four suites, whose sizes are given in Table~\ref{tab-exp}, as follows: 
\begin{itemize}
\item MEM-SAT and MEM-ENT are satisfiability resp. entailment problems
	generated by the verification of programs that are building blocks of heap-list based memory allocators, including:
	create a heap-list with one element, 
	split a memory chunk into two consecutive memory chunks, 
	join two memory chunks, 
	search a memory chunk of size bigger than a given parameter (our running example), 
	or search an address inside a heap-list.
\item RANDOM-SAT and RANDOM-ENT are satisfiability resp. entailment problems
	which are randomly or manually generated.
	Starting from the path and verification conditions for programs manipulating heap lists,
	we replace some $\hls{}$ atoms with their unfoldings in order to generate
	formulas with more spatial atoms. 
	RANDOM-SAT includes also formulas where the atoms and their start
	and end address terms are generated randomly. In addition, we generate some benchmarks manually. 
	This suite is motivated by testing the scalability of \cspenp.
\end{itemize}
%
%

\smallskip
\mypar{Experiments.}
We run \cspenp over the four benchmark suites, using a Ubuntu-16.04 64-bit lap-top with an Intel Core i5-8250U CPU and 2GB RAM.
The experimental results are summarized in Table~\ref{tab-exp}. We set the timeout to 60 seconds. The statistics of average time and maximum time do not include the time of timeout instances.
To the best of our knowledge, there have not been solvers that are capable of solving the \slah\ formulas that include, points-to, block, and $\hls{}$ atoms. 
	The solver SLar~\cite{KimuraT17} was designed to solve entailment problems involving points-to, block, and $\ls{}$ atoms. Nevertheless, we are unable to find a way to access SLar, thus failing to compare with it on ASL formulas. 
	Moreover, the examples used by SLar, available online, are in a format that seems nontrivial to translate into the SL-COMP format.

\begin{table}[htbp]
		\vspace{-2eX}
		\caption{Experimental results, time measured in seconds}
		\label{tab-exp}
		\vspace{-2eX}
\begin{center}
		\setlength{\tabcolsep}{1eX}
		\renewcommand{\arraystretch}{1.2}
		\begin{tabular}{| c | c | c | c | c | c |}
			\hline
			Benchmark suite & $\#$instances & Timeout & Avg. time & Min. time & Max. time \\ 
			\hline
			\hline
			MEM-SAT & 38 & 0 & 0.05 & 0.03 &0.12 \\
			\hline
			RANDOM-SAT & 50 & 0 & 0.09 & 0.02 & 0.53 \\
			 \hline
			 TOTAL & 88 & 0 & 0.07 & 0.02 & 0.53 \\
			 \hline
			 \hline
			MEM-ENT  & 43 & 0 & 3.05  & 0.34  & 9.98 \\ 
			\hline
			 RANDOM-ENT & 59 & 2 & 13.39 & 0.04 & 48.85 \\
			 \hline
			 TOTAL & 102 & 2 & 8.94 & 0.04 & 48.85 \\
			 \hline
		\end{tabular}
\end{center}
		\vspace{-6eX}
	\end{table}

As expected, solving entailment instances is more expensive than solving satisfiability instances.
We recall from Section~\ref{sec:ent} that the procedure for entailment queries satisfiability of several formulas.
\cspenp\ efficiently solves the benchmark instances originated from program's verification, 
namely MEM-SAT and MEM-ENT, with the average time in 0.05 and 3.05 seconds respectively. 


Table~\ref{tab-exp} shows that some entailment instances 
are challenging for \cspenp. 
For instance, the maximum time in MEM-ENT suite is 
9.98, and there are 2 timeout instances in RANDOM-ENT suite (more than 2 min). 
By inspecting these challenging instances, we found that
(i) they require splitting some spatial atoms ($\blk{}$ or  $\hls{}$) in the antecedent, 
	which is potentially time-consuming, and
(ii) they correspond to valid entailment problems 
	where \cspenp\ has to explore all the total preorders,
	which is time-consuming. 
We noticed that when the entailment problem is invalid, 
the heuristics implemented are able to quickly 
find some total preorder under which the entailment does not hold.



\section{Conclusion}
\label{sec:conc-hls}

In this work, we investigated \slah, a separation logic fragment that allows pointer arithmetic inside inductive definitions so that the commonly used data structures e.g. heap lists can be defined.
We show that the satisfiability problem of {\slah} is NP-complete and 
    the entailment problem is coNP-complete.
We implemented the decision procedures in a solver, \cspenp, and use it 
   to efficiently solve more than hundred problems issued from verification of 
   program manipulating heap lists
   or randomly generated problems. 
For future work, it is interesting to see whether the logic \slah\ and its decision procedures can be extended to specify free lists, another common data structure in memory allocators \cite{Knuth97} in addition to heap lists. Moreover, since bit operations are also widely used in memory allocators, it would also be interesting to automate the reasoning about bit operations inside inductive definitions.

\bibliographystyle{splncs04}
\bibliography{biblio}

\newpage
\appendix


\section{{\qfpa} abstraction for \slah\ formulas}\label{app:sat-hls}

\subsection{{\qfpa} summary of $\hls{}$ atoms}
\label{ssec:sat-hls-abs}

\noindent {\bf Lemma~\ref{lem-hls}}.
\emph{Let $ \hls{}(x, y; z)$ be an atom in \slah\
representing a non-empty heap, where $x, y, z$ are three distinct variables in $\cV$.
Then there is an {\qfpa} formula, denoted by $\abs^+(\hls{}(x,y; z))$, which summarizes $\hls{}(x, y; z)$, namely for every stack $s$ $s \models \abs^+(\hls{}(x,y; z))$ iff there exists a heap $h$ such that $s, h \models \hls{}(x, y, z)$. }

\begin{proof}
The constraint that the atom represents a non-empty heap means that 
the inductive rule defining $\hls{}$ in Equation~(\ref{eq:hlsv-rec})
should be applied at least once. Notice that the semantics of this rule
defines, at each inductive step,
a memory block starting at $x$ and ending before $x'$ of size $x'-x$.
By induction on $k \ge 1$, we obtain that $\hls{k}(x, y; z)$ defines
a memory block of length $y-x$ such that 
$2 k \le y-x \le k z$. Then $\hls{}(x, y; z)$ is summarized by the formula $\exists k.\ k \ge 1 \wedge 2k \le y -x \le kz$, which is a non-linear arithmetic formula.

The formula $\exists k.\ k \ge 1 \wedge 2k \le y -x \le kz$ is actually equivalent to the disjunction of the two formulas corresponding to the following two  cases:
\begin{itemize}
\item If $2 = z$,  then $\abs^+(\hls{}(x, y; z))$ has the formula $\exists k.\ k \ge 1 \land y -x = 2k$, which is equivalent to the {\qfpa} formula $2 \le y -x \wedge y - x \equiv_2 0$, as a disjunct.
\item If $2 < z$, then we consider the following sub-cases:  
\begin{itemize}
\item if $k = 1$ 
	then $\abs^+(\hls{}(x, y; z))$ contains the formula
		$ 2 \leq y-x \le z$ as a disjunct;
		
\item if $k \ge 2$, then we observe that the intervals 
	$[2k, k z]$ and $[2(k+1), (k+1) z]$
	overlap.
	Indeed, 
\begin{eqnarray*}
k z - 2(k+1) & = & 
k (z -2) - 2 
 \ge k - 2 \ge 0.
\end{eqnarray*}
    Therefore, $\bigcup \limits_{k \ge 2} [2k, kz] = [4, \infty)$. It follows that the formula 
    $\exists k.\ k \ge 2 \land 2k \le y-x \le k z$ 
    is equivalent to $4 \le y-x$. Therefore, $\abs^+(\hls{}(x, y; z))$ contains $4 \le y-x$ as a disjunct.
\end{itemize}
\end{itemize}
To sum up, we obtain
\begin{eqnarray*}
\abs^+(\hls{}(x, y; z)) & \triangleq 
          & \big(2 = z 
				 ~\land~ \exists k.\ k \ge 1 \land 2k = y-x \big) \\
   & \lor & \big(2 < z
            ~\land~
                   \big(2 \leq y-x \le z \lor 4 \le y-x 
                          \big) 
            \big),
\end{eqnarray*}
which can be further simplified into 
\begin{eqnarray*}
\abs^+(\hls{}(x, y; z)) & \triangleq 
          & \big(2 = z \land 2 \le y -x \wedge y - x \equiv_2 0 \big) \\
   & \lor & \big(2 < z 
            \land 2 \le y - x
                          \big).
\end{eqnarray*}
%
%
%
\qed
\end{proof}

\subsection{$\abs(\varphi)$: The {\qfpa} abstraction of $\varphi$}

We utilize $\abs^+(\hls{}(x, y; z))$ 
to obtain in polynomial time an equi-satisfiable {\qfpa} abstraction for a symbolic heap $\varphi$, denoted by $\abs(\varphi)$.

We introduce some notations first.
Given a formula $\varphi\equiv \Pi : \Sigma$, 
$\atoms(\varphi)$ denotes the set of spatial atoms in $\Sigma$, and $\patoms(\varphi)$ denotes the set of predicate atoms in $\Sigma$. 

\begin{definition}{(Presburger abstraction of \slah\ formula)}
Let $\varphi\equiv \Pi : \Sigma$ be a \slah\ formula.
The abstraction of $\varphi\equiv \Pi : \Sigma$, 
	denoted by $\abs(\varphi)$ is 
	the formula $\Pi\wedge\phi_{\Sigma}\wedge\phi_*$
	where:
\begin{itemize}
\item $\phi_{\Sigma}\triangleq\bigwedge\limits_{a \in \atoms(\varphi)}\abs(a)$ such that
{\small
\begin{eqnarray}
\abs(t_1\mapsto t_2) & \triangleq & \ltrue \\
\abs(\blk(t_1,t_2))  & \triangleq & t_1<t_2 \\
\abs(\hls{}(t_1, t_2, t_3))  & \triangleq & t_1=t_2 \\
& \lor &(t_1 < t_2 \land \abs^+ (\hls{}(x, y;z))[t_1/x, t_2/y, t_3/z]).
\end{eqnarray}
}

\item $\phi_\sepc\triangleq\phi_1\wedge \phi_2 \wedge \phi_3$ 
	specifies the semantics of separating conjunction, 
where 	
{\small
\begin{eqnarray}
\phi_1 & \triangleq & \bigwedge\limits_{a_i,a_j \in \patoms(\varphi), i<j} 
	\begin{array}[t]{l}
	(\isnonemp_{a_i} \land \isnonemp_{a_j}) \limp \\
	\quad (\atomtail(a_j) \le \atomhead(a_i)\lor \atomtail(a_i) \le \atomhead(a_j))
	\end{array}
\\
\phi_2 & \triangleq & \bigwedge\limits_{a_i \in \patoms, a_j \not \in \patoms} 
	\begin{array}[t]{l}
	(\isnonemp_{a_i}) \limp \\
	\quad (\atomtail(a_j) \le \atomhead(a_i) \lor \atomtail(a_i) \le \atomhead(a_j))
	\end{array}
\\
\phi_3 & \triangleq & \bigwedge\limits_{a_i, a_j \not \in \patoms(\varphi),  i < j} 
		\atomtail(a_j) \le \atomhead(a_i) \lor \atomtail(a_i) \le \atomhead(a_j)
\end{eqnarray}
}
\end{itemize}
and for each spatial atom $a_i$, $\isnonemp_{a_i}$ is an abbreviation of the formula $\atomhead(a_i) < \atomtail(a_i)$.
\end{definition}

For formulas $\varphi\equiv \exists \vec{z}\cdot \Pi : \Sigma$,
we define $\abs(\varphi) \triangleq \abs(\Pi:\Sigma)$ since $\exists \vec{z}\cdot \Pi : \Sigma$ and $\Pi:\Sigma$ are equi-satisfiable.

\vspace{4mm}

\noindent {\bf Proposition~\ref{prop-sat-correct}}.
\emph{A \slah\ formula $\varphi$ is satisfiable iff $\abs(\varphi)$ is satisfiable.}

\smallskip

\begin{proof}
\emph{``Only if'' direction}:
Suppose that $\varphi$ is satisfiable. 

Then there exists a stack $s$ and a heap $h$
such that $s,h \models \varphi$. 
We show $s \models \abs(\varphi)$.
The semantics of \slah\ implies that at least one disjunct in $\abs(a)$
is true for every predicate atom $a$. Therefore, $s \models \phi_\Sigma$. Moreover, $s \models \phi_\sepc$, as a result of the semantics of separating conjunction. Thus $s \models \abs(\varphi)$ and  $\abs(\varphi)$ is satisfiable.

\smallskip
\noindent \emph{``If'' direction}: Suppose that $\abs(\varphi)$ is satisfiable. 

Then there is an interpretation $s$ such that $s \models \abs(\varphi)$.
We build a heap $h$ from $s$ and $\varphi$ as follows: 
\begin{itemize}
\item For each points-to atom $t_1 \pto t_2$ in $\varphi$, $h(s(t_1)) = s(t_2)$.
\item For each block atom $\blk(t_1, t_2)$ in $\varphi$, $h(n) = 1$ for each $n \in [s(t_1), s(t_2)-1]$.
\item For every predicate atom $\hls{}(t_1,t_2; t_3)$ in $\varphi$, we have $s \models \abs(\hls{}(t_1,t_2; t_3))$.  Then either $s(t_1) = s(t_2)$ or $s(t_1) < s(t_2)$ and $s \models \abs^+(\hls{}(x, y; z))[t_1/x, t_2/y, t_3/z]$. 
\begin{itemize}
\item For the first case, we let $h(s(t_1))$ undefined. 

\item For the second case,  $s \models (2 = t_3 \wedge t_1 < t_2 \wedge t_2 - t_1 \equiv 0 \bmod 2) \vee (2 < t_3 \wedge 2 \le t_2 - t_1)$. 
\begin{itemize}
\item If $s(t_3) = 2$, then let $h(s(t_1) + 2(i-1)) = 2$ and $h(s(t_1) + 2i-1) = 1$ for every $i: 1 \le i \le \frac{s(t_2) - s(t_1)}{2}$.
\item If $s(t_3) > 2$, then from $s(t_2) - s(t_1) \ge 2$, we know that there is a sequence of numbers $n_1, \cdots, n_\ell$ such that $2 \le n_i \le s(t_3)$ for every $i \in [\ell]$ and $s(t_2) = s(t_1) + \sum \limits_{i \in [\ell]} n_i$.  We then let $h(s(t_1) + \sum \limits_{j \in [i-1]} n_j) = n_i$ for every $i \in [\ell]$, and let $h(n') = 1$ for all the other addresses $n' \in [s(t_1), s(t_2)-1]$. 
\end{itemize}
From the construction above, we know that the subheap of the domain $[s(t_1), s(t_2)-1]$ satisfies $\hls{}(t_1, t_2; t_3)$.
\end{itemize}
\end{itemize}
From the definition of $h$, we know that $s, h \models \varphi$. Therefore, $\varphi$ is satisfiable.
\qed
\end{proof}


\begin{remark}
From the definition of $\abs(\varphi)$, it follows that 
$\abs(\varphi)$ is in {\qfpa} and 
the size of $\abs(\varphi)$ is polynomial in that of $\varphi$. 
From the fact that the satisfiability of {\qfpa} is in NP, we conclude that the satisfiability of \slah\ is in NP.
%
\end{remark}

\section{Proof of Lemma~\ref{lem-eub}}

\smallskip

\noindent {\bf Lemma~\ref{lem-eub}}.
\emph{For an {\slah} formula $\varphi \equiv \Pi: \hls{}(t'_1, t'_2; t'_3)$, a {\qfpa} formula $\xi_{eub,\varphi}(z)$ can be constructed in linear time such that for every store $s$ satisfying $s \models \abs(\varphi)$, we have $s[z \gets \eub_\varphi(s)] \models \xi_{eub,\varphi}(z)$ and $s[z \gets n] \not \models \xi_{eub,\varphi}(z)$ for all $n \neq \eub_\varphi(s)$. 
}

\smallskip

\begin{proof}
From the construction of $\abs^+(\hls{}(t'_1, t'_2; t'_3))$ in Section~\ref{sec:sat-hls}, we know that for every store $s$, there is a heap $h$ such that $s, h \models \Pi: \hls{}(t'_1, t'_2; t'_3)$ iff $s \models \Pi \wedge \big((t'_3 = 2 \wedge 2 \le t'_2 - t'_1 \wedge t'_2 - t'_1 \equiv_2 0) \vee ( 2 < t'_3 \wedge 2 \le t'_2 - t'_1)\big)$. Then the fact that $z$ appears in some unfolding scheme of $\varphi$ w.r.t. some store $s$ 
 can be specified by the formula 
\[
\Pi \wedge 
\left(
\begin{array}{l}
(t'_3 = 2 \wedge z = 2 \wedge 2 \le t'_2 - t'_1 \wedge t'_2 - t'_1 \equiv_2 0)\ \vee \\
(2 < t'_3 \wedge 2 \le z \le t'_3\ \wedge (2 \le t'_2 - t'_1 - z \vee t'_2 - t'_1 = z))
\end{array}
\right),
\]
where $t'_2 - t'_1 - z$ denotes the remaining size of the heap after removing a memory chunk of size $z$. Therefore, we define $\xi_{eub,\varphi}(z)$ as 
\[
\begin{array}{l}
 \Pi \wedge
\left(
\begin{array}{l}
(t'_3 = 2 \wedge z = 2 \wedge 2 \le t'_2 - t'_1 \wedge t'_2 - t'_1 \equiv_2 0)\ \vee  \\
 \left(
 \begin{array}{l}
 2 < t'_3 \wedge  2 \le z \le t'_3 \wedge (2 \le t'_2 - t'_1 - z \vee t'_2 - t'_1 = z) \ \wedge  \\
 \forall z'.\ z <  z' \le  t'_3 \rightarrow \neg (2 \le t'_2 - t'_1 - z' \vee t'_2 - t'_1 = z') 
\end{array}
 \right)
\end{array}
 \right),
\end{array}
\]
where $ \forall z'.\ z <  z' \le  t'_3 \rightarrow \neg  (2 \le t'_2 - t'_1 - z' \vee t'_2 - t'_1 = z') $ asserts that $z$ is the maximum chunk size upper bound among unfolding schemes of $\varphi$. The formula $ \forall z'.\ z <  z' \le  t'_3 \rightarrow \neg  (2 \le t'_2 - t'_1 - z' \vee t'_2 - t'_1 = z') $ is equivalent to $ \forall z'.\ z <  z' \le  t'_3 \rightarrow   (t'_2 - t'_1 - z' \le 1 \wedge t'_2 - t'_1 \neq z')$, and can be simplified into  
$$
\begin{array}{l}
z+1 \le t'_3 \rightarrow  (t'_2 - t'_1 \le z+2 \wedge t'_2 - t'_1  \neq z+1)\ \wedge \\
z+2 \le t'_3 \rightarrow t'_2 - t'_1  \neq z+2,
\end{array}
$$
where $t'_2 - t'_1  \neq z+1$ is an abbreviation of $t'_2 - t'_1 < z+1 \vee z+1 < t'_2 - t'_1$, similarly for $t'_2 - t'_1  \neq z+2$.

It follows that $\xi_{eub,\varphi}(z) $ can be simplified into 
\[
 \Pi  \wedge 
 \left(
\begin{array}{l}
(t'_3 = 2 \wedge z = 2 \wedge 2 \le t'_2 - t'_1 \wedge t'_2 - t'_1 \equiv_2 0)\ \vee  \\
 \left(
\begin{array}{l}
 2 < t'_3 \wedge  2 \le z \le t'_3 \wedge (2 \le t'_2 - t'_1 - z \vee t'_2 - t'_1 = z) \ \wedge  \\
z+1 \le t'_3 \rightarrow  (t'_2 - t'_1 \le z+2 \wedge t'_2 - t'_1  \neq z+1)\ \wedge \\
z+2 \le t'_3 \rightarrow t'_2 - t'_1  \neq z+2
\end{array}
\right)
\end{array}
 \right),
\]
which is a {\qfpa} formula. 
\qed
\end{proof}

\section{Proof of Lemma~\ref{lem-hls-hls}}

\noindent {\bf Lemma~\ref{lem-hls-hls}}
\emph{
Let $\varphi \equiv \Pi: \hls{}(t'_1, t'_2; t'_3)$, $\psi \equiv \hls{}(t_1, t_2; t_3)$, and $\preceq$ be a total preorder over $\addr(\varphi) \cup \addr(\psi)$ such that $C_\preceq \models t'_1 < t'_2 \wedge t'_1 = t_1 \wedge t'_2 = t_2$. 
Then $\varphi \models_\preceq \psi$ iff $C_\preceq \wedge \abs(\varphi)  \models  
\forall z.\ \xi_{eub, \varphi}(z) \rightarrow z \le t_3$.  
}
\smallskip

\begin{proof}
\noindent \emph{``Only if'' direction}: Suppose $\varphi  \models_\preceq \psi$.  
%

At first, we observe that $C_\preceq \wedge \abs(\varphi)$ and $\varphi$ contain the same number of variables. 

Let $s$ be an interpretation such that $s \models C_\preceq \wedge \abs(\varphi)$. From $C_\preceq \models t'_1 < t'_2 \wedge t'_1 = t_1 \wedge t'_2 = t_2$, we have $s(t'_1) < s(t'_2)$, $s(t'_1) = s(t_1)$, and $s(t'_2) = s(t_2)$. 

At first, from $s \models \abs(\varphi)$ and Lemma~\ref{lem-eub}, we deduce that $s[z \rightarrow \eub_\varphi(s)] \models \xi_{eub,\varphi}(z)$.
Moreover, from the definition of $\eub_\varphi(s)$, we know that there is an unfolding scheme, say $sz_1, \cdots, sz_\ell$, such that $\eub_\varphi(s) = \max(sz_1, \cdots, sz_\ell)$. From the definition of unfolding schemes, we have $2 \le sz_i \le s(t'_3)$ for every $i \in [\ell]$ and $s(t'_2) = s(t'_1) + \sum \limits_{i \in [\ell]} sz_i$. Therefore, there is a heap $h$ such that $s, h \models \varphi$ and for every $i \in [\ell]$, $h(s(t'_1) + \sum \limits_{j \in [i-1]} sz_j) = sz_i$. From the assumption $\varphi  \models_\preceq \psi$, we deduce that $s, h \models \psi \equiv \hls{}(t_1, t_2; t_3)$. Then each chunk of $h$ has to be matched exactly to one unfolding of $ \hls{}(t_1, t_2; t_3)$. Thus, $2 \le sz_i \le s(t_3)$ for every $i \in [\ell]$. Therefore, $\eub_\varphi(s) =  \max(sz_1, \cdots, sz_\ell) \le s(t_3)$. We deduce that $s[z \gets \eub_\varphi(s)] \models \xi_{eub,\varphi}(z) \wedge z \le t_3$. Moreover, from Lemma~\ref{lem-eub}, we know that for all $n \neq \eub_\varphi(s)$, $s[z \gets n] \not \models \xi_{eub,\varphi}(z)$. Consequently, $s \models \forall z.\  \xi_{eub,\varphi}(z) \rightarrow z \le t_3$.
 
\medskip
\noindent \emph{``If'' direction}: Suppose $C_\preceq \wedge \abs(\varphi)  \models  \forall z.\ \xi_{eub, \varphi}(z) \rightarrow z \le t_3$ and $s, h \models C_\preceq \wedge \Pi: \hls{}(t'_1, t'_2; t'_3)$. We show $s, h \models \psi$.

From $s, h \models C_\preceq \wedge \Pi: \hls{}(t'_1, t'_2; t'_3)$, we know that there are $sz_1, \cdots, sz_\ell$ such that they are the sequence of chunk sizes in $h$. Therefore, $(sz_1, \cdots, sz_\ell)$ is an unfolding scheme of $\varphi$ w.r.t. $s$.  From the definition of $\eub_\varphi(s)$, we have $\eub_\varphi(s) $ is the maximum chunk size upper bound of the unfolding schemes of $\varphi$ w.r.t. $s$. Therefore, $\max(sz_1, \cdots, sz_\ell) \le \eub_\varphi(s)$.

From $s, h \models C_\preceq \wedge \Pi: \hls{}(t'_1, t'_2; t'_3)$, we deduce that $s \models C_\preceq \wedge \abs(\varphi)$. 
Because $C_\preceq \wedge \abs(\varphi)  \models  \forall z.\ \xi_{eub, \varphi}(z) \rightarrow z \le t_3$, we have $s \models \forall z.\ \xi_{eub, \varphi}(z) \rightarrow z \le t_3$. 
From Lemma~\ref{lem-eub}, we know that $s[z \gets \eub_\varphi(s)] \models \xi_{eub, \varphi}(z)$. 
Therefore, we deduce that $s[z \gets \eub_\varphi(s)] \models z \le t_3$, namely, $\eub_\varphi(s) \le s(t_3)$. Then $\max(sz_1, \cdots, sz_\ell) \le \eub_\varphi(s) \le s(t_3)$. It follows that for every $i \in [\ell]$, $2 \le sz_i \le s(t_3)$. We conclude that $s, h \models \hls{}(t_1, t_2; t_3) \equiv \psi$.
\qed
\end{proof}


\section{Complexity of the entailment problem}

Theorem~\ref{thm-entail} states that the upper bound of the entailment problem is \textsc{EXPTIME}. 
The lower bound is given by the complexity of the entailment problem for quantifier free formulas in \textsc{ASL}, which is shown to be \textsc{coNP} in~\cite{BrotherstonGK17}.

\textbf{Claim:} The entailment problem in \slah\ is in \textbf{coNP}.

\begin{proof}

Let fix $\varphi \models \psi$ an entailment problem in \slah.

We build in polynomial time a \PbA\ formula $\chi$ from $\varphi$ and $\psi$ 
such that 
for any stack $s$, $s\models \chi$ 
			   iff $\exists h.~s,h\models \varphi$ and $s,h \not\models \psi$.

We suppose that both formula $\varphi$ and $\psi$ have been transformed in normal form as follows:
\begin{itemize}
\item atoms $\hls{}(x,y;v')$ with $\abs(\varphi)\models (v'<2 \lor x=y)$ have been removed and replaced by $x=y$\mihaela{use $\xi$ ?}
\item atoms $\hls{}(x,y;v')$ with $\abs(\varphi)\models y-x=2 \land v' \ge 2$ have been replaced by $\exists z.~ z \le v' \land y-x=2 : ~x\pto z \star \blk(x+1,y)$
\item ...\mihaela{TODO}
\end{itemize}
The normalization can be done in a linear number of calls to the \EPbA\ procedure.

The formula $\chi$ shall ensure that $\varphi$ is satisfiable, so
it contains as conjunct $\abs(\varphi)$.

The entailment may be invalid iff one of the following holds:
\begin{enumerate}
\item $\psi$ is not satisfiable, i.e., $\abs(\psi)$ is not satisfiable.

\item there exists an address inside a spatial atom of $\varphi$ 
	which is not inside a spatial atom of $\psi$; 
	we denote this formula by $cov(\varphi,\psi)$;

\item similarly, $cov(\psi,\varphi)$ (the semantics is non intuitionistic),

\item a points-to or a heap-list atom of $\psi$ ``covers'' a domain of addresses
	which is inside the one of a block atom of $\varphi$;
	we denote this formula by $nblk(\varphi,\psi)$;
	
\item a points-to atom of $\psi$ defines an address inside (strictly) 
	a heap-list atom of $\varphi$;
	we denote this formula by $npto(\varphi,\psi)$;
	
\item a points-to atom $x\pto v$ of $\psi$ defines the same address as the start address
	of a heap-list atom $\hls{}(x,y;v')$ in $\varphi$;
	we denote this formula by $npto'(\varphi,\psi)$;
	(this is true because we are in the normal form);
	
\item a points-to atoms $x\pto v$ of $\psi$ defines the same address as a 
	points-to atom $x\pto v'$ of $\varphi$ but $v\neq v'$;
		
\item a heap-list atom $\hls{}(x,y;v')$ of $\psi$ starts at an address at which is defined a points-to atom $x\pto v$ of $\varphi$ such that $v> v'$;
			
\item a heap-list atom $\hls{}(x,y;v')$ of $\psi$ ``covers'' the domain of
	addresses if an atom $\hls{}(z,w;v)$ of $\varphi$ but $v' < v$;

\item a heap-list atom $\hls{}(x,y;v')$ of $\psi$ ``covers'' a domain of addresses
	that can not be filded into a $\hls{}$ atom with chunks of maximal size less or equal to $v'$\mihaela{difficult in P!}

\end{enumerate}

\end{proof}

\end{document}